\documentclass[11pt,letterpaper]{article}
\usepackage[margin=1in]{geometry}
\usepackage{amsmath,amssymb,amsthm,mathtools,microtype}
\usepackage{booktabs}
\usepackage[T1]{fontenc}
\usepackage{lmodern}
\usepackage{enumitem}
\usepackage{tikz}
\usetikzlibrary{decorations.pathreplacing,arrows.meta}
\definecolor{setone}{HTML}{88B4E7}
\definecolor{settwo}{HTML}{F4A261}
\definecolor{setthree}{HTML}{9FD8A8}
\definecolor{setonedark}{HTML}{2A63B0}
\definecolor{settwodark}{HTML}{B5541C}
\definecolor{setthreedark}{HTML}{2E7D4F}
\definecolor{coredark}{HTML}{6B3FA0}
\definecolor{bnddark}{HTML}{8A5A00}
\usepackage{xurl}
\usepackage[hidelinks]{hyperref}
\hypersetup{
  pdftitle={PPAD-Completeness of Contiguous Envy-Free Cake Cutting with Additive Valuations},
  pdfauthor={Yasushi Kawase, Bodhayan Roy, Mohammad Azharuddin Sanpui},
  pdfkeywords={cake cutting, envy-freeness, connected allocation, PPAD, additive valuations}
}
\newtheorem{theorem}{Theorem}
\newtheorem{lemma}[theorem]{Lemma}
\newtheorem{corollary}[theorem]{Corollary}
\newtheorem{proposition}[theorem]{Proposition}
\newcommand{\clip}{\operatorname{clip}}

\newcommand{\Cake}{\textnormal{\textsc{Contiguous-EF}}}
\newcommand{\TG}{\textnormal{\textsc{ThresholdGameNE}}}
\newcommand{\email}[1]{\href{mailto:#1}{\nolinkurl{#1}}}

\title{PPAD-Completeness of Contiguous Envy-Free\protect\\
Cake Cutting with Additive Valuations\thanks{This work was supported by JSPS KAKENHI Grant Number JP25K00137 and JST ERATO Grant Number JPMJER2301.}}
\author{Yasushi Kawase\thanks{Chuo University, Tokyo, Japan (\email{ykawase422@g.chuo-u.ac.jp}).}
\and Bodhayan Roy\thanks{Indian Institute of Technology Kharagpur, Kharagpur, India (\email{broy@maths.iitkgp.ac.in}).}
\and Mohammad Azharuddin Sanpui\thanks{Indian Institute of Technology Kharagpur, Kharagpur, India (\email{azharuddinsanpui123@gmail.com}).}}
\date{}

\begin{document}
\pagenumbering{roman}
\maketitle
\thispagestyle{empty}
\begin{abstract}
We prove that computing a complete contiguous $\varepsilon$-envy-free cake allocation is PPAD-complete for normalized rational piecewise-constant additive valuations.
The hardness holds even when $\varepsilon$ is given in unary, and even when every density takes only the values $\lambda$ and $2\lambda$ for some $\lambda>0$ common to all agents and has at most six breakpoints.
As a consequence, computing an exact contiguous envy-free allocation is also PPAD-complete, even under the same density restrictions.
Our key ingredient is to assign prices to the pieces so that, in every approximately envy-free allocation, all prices are nearly equal, regardless of how the pieces are assigned to the agents.
We further show that, for every fixed integer $c\ge1$, computing a contiguous allocation of indivisible goods on a path that is envy-free up to $c$ goods (EF$c$) is PPAD-complete for additive valuations, and this remains true even for $c=\lfloor m^\alpha\rfloor$ with any fixed rational $0<\alpha<1$, where $m$ is the number of goods.
This hardness holds even for binary valuations in which each good is valued by exactly one agent and all agents value the same number of goods.
All these results extend to chores.
Our results answer three open questions at once: whether contiguous envy-free cake cutting is PPAD-hard for additive valuations, the complexity of computing contiguous EF1 and EF2 allocations of goods on a path, and the corresponding question for EF1 allocations of chores.
\end{abstract}

\newpage
\pagenumbering{arabic}

\section{Introduction}\label{sec:intro}
Fair division asks how to allocate a heterogeneous resource among agents with different preferences.
Envy-freeness is a natural fairness requirement: each agent should value its own share at least as much as any other agent's share.
In applications such as dividing land or assigning time slots, a share consisting of scattered parts may be of little use, and thus it is natural to require every agent to receive one connected piece.

In cake cutting, the resource is the interval $[0,1]$, and each of $n$ agents has a normalized additive valuation $V_i$, given here by an explicit rational piecewise-constant density.
A complete contiguous allocation partitions the cake into $n$ intervals and gives one interval to each agent.
The assignment of the intervals to the agents is part of the solution.
Such an allocation $(X_1,\ldots,X_n)$, where $X_i$ is the piece of agent $i$, is \emph{$\varepsilon$-envy-free} if
\[
 V_i(X_i)\ge V_i(X_j)-\varepsilon\qquad\text{for all }i,j\in[n].
\]
For $\varepsilon=0$, we simply say that the allocation is \emph{envy-free}.
Contiguous envy-free allocations always exist~\cite{Stromquist1980,Su1999}, but no finite protocol in the Robertson--Webb model finds one for three or more agents~\cite{Stromquist2008}.
For explicitly given valuations, computing an $\varepsilon$-envy-free allocation lies in PPAD~\cite{DengQS2012}.
PPAD-hardness, however, was known only for non-additive preferences~\cite{DengQS2012,GaoRRS2024,HollenderR2025}, and it remained open for additive valuations, the standard model of cake cutting, even for explicit piecewise-constant densities~\cite{GoldbergHS2020b}.

A closely related question arises for indivisible goods.
Here, $m$ goods are arranged on a path, for example, as consecutive time slots, and each agent has an additive valuation over the goods.
A contiguous allocation gives each agent a possibly empty set of consecutive goods, and as for the cake, the assignment of these sets to the agents is part of the solution.
An envy-free allocation need not exist, for example, when two agents value a single good, and thus one asks for \emph{envy-freeness up to $c$ goods} (EF$c$): removing at most $c$ goods from any envied bundle eliminates the envy.
Contiguous EF1 allocations always exist~\cite{Igarashi2023}.
Bil\`o et al.~\cite{Bilo+2019} posed the complexity of computing contiguous EF1 or EF2 allocations as an open problem.
For chores, which the agents want to avoid, contiguous EF1 allocations also exist~\cite{SunL2026,BiloLV2026}, and Bil\`o, Loebl, and Vinci~\cite{BiloLV2026} conjectured that computing them is PPAD-complete.

\subsection{Our Results}
We answer these questions.

\begin{theorem}\label{thm:main}
Computing a contiguous $\varepsilon$-envy-free allocation for normalized rational piecewise-constant additive valuations is PPAD-complete.
The hardness holds even when $\varepsilon$ is given in unary, and even when every density takes only the values $\lambda$ and $2\lambda$ for some $\lambda>0$ common to all agents and has at most six breakpoints.
\end{theorem}

Theorem~\ref{thm:main} holds for both encodings of $\varepsilon$: hardness for unary $\varepsilon$ implies hardness for binary $\varepsilon$, and membership holds even for binary $\varepsilon$ by Lemma~\ref{lem:membership}.
In the hard instances of Theorem~\ref{thm:main}, the number $n$ of agents is part of the input, and $\varepsilon=\Theta(n^{-4})$.
Thus, unless $\mathrm{PPAD}\subseteq\mathrm{FP}$, no algorithm computes such allocations in time polynomial in the input length and $1/\varepsilon$, even under the stated density restrictions.
These restricted instances nevertheless admit an efficient polynomial-time approximation scheme for additive envy, as explained in Section~\ref{sec:conclusion}.
Since an envy-free allocation is $\varepsilon$-envy-free for every $\varepsilon>0$, Theorem~\ref{thm:main} and Lemma~\ref{lem:membership} also show that computing an exact contiguous envy-free allocation is PPAD-complete (Corollary~\ref{cor:exact}).

To transfer this hardness to indivisible goods, we place goods at the quantiles of each agent's valuation, each valued only by that agent, so that every contiguous EF$c$ allocation of the goods yields an approximately envy-free allocation of the cake.
Membership follows by rounding an exact envy-free division of an auxiliary cake with piecewise-linear virtual valuations~\cite{BiloLV2026,DupreLaTourI2026,FilosRatsikasHHH2026}.
We write EF$c_{\mathrm{outer}}$ for the stronger notion of Bil\`o et al.~\cite{Bilo+2019}, which requires the removed goods to leave the envied bundle connected or empty.

\begin{theorem}\label{thm:discrete}
For every fixed integer $c\ge1$, computing a contiguous EF$c$ allocation of explicitly given goods on a path is PPAD-complete for nonnegative rational additive valuations.
The hardness holds even for binary valuations in which every good is valued by exactly one agent and all agents value the same number of goods.
Both statements also hold for EF$c_{\mathrm{outer}}$ in place of EF$c$, and for $c=\lfloor m^\alpha\rfloor$ for every fixed rational $0<\alpha<1$, where $m$ is the number of goods.
\end{theorem}

Theorems~\ref{thm:main} and~\ref{thm:discrete} and Corollary~\ref{cor:exact} extend to chores, where each agent prefers pieces of small cost, under the same restrictions (Proposition~\ref{prop:chores}).

\subsection{Overview of the Reduction}\label{sec:overview}
The additive case is difficult for two reasons: the value of a piece cannot depend on how the rest of the cake is divided, unlike the preferences in the hardness proof of Deng, Qi, and Saberi~\cite{DengQS2012}, and a solution of the source problem must be recovered from every approximately envy-free allocation, whatever the assignment of the pieces to the agents.

We reduce from finding an approximate equilibrium of a threshold game, which is PPAD-hard~\cite{DeligkasFHM2024}.
In such a game, each player $v$ chooses a value $z_v\in[0,1]$ and prefers $1$ if the total value of its in-neighbors is below $1/2$ and $0$ if it is above $1/2$.
We work with an unnormalized cake $[0,NB]$, which we divide into $N$ slots of length $B=\Theta(N)$, one for each agent, and we represent each $z_v$ by the position of one cut.
Since every agent values each piece at nearly its length, approximate envy-freeness places the $j$th piece $P_j$ near the $j$th slot (Lemma~\ref{lem:localization}), but it determines the lengths too coarsely to decode $z$.

Our key ingredient is to give each piece a price.
For disjoint sets $S_1,\ldots,S_s$ of indices, one for each player, let $q_j=\ell_j+z_v$ if $j\in S_v$ and $q_j=\ell_j$ otherwise, where $\ell_j$ is the length of $P_j$.
We choose the valuations so that, with utility given by value minus price, every agent $i$ weakly prefers both the piece $P_i$ in its own slot and the piece $P_{\pi(i)}$ to every other piece, where $\pi$ is a fixed cyclic order of all indices (Lemma~\ref{lem:envelope}).
If agent $i$ receives $P_{\sigma(i)}$ for a permutation $\sigma$ in an $\eta$-envy-free allocation, then $q_{\sigma(i)}\ge\max\{q_i,q_{\pi(i)}\}-\eta$, and since $\sum_iq_{\sigma(i)}=\sum_iq_i$, summing over all agents gives
\[
 \sum_{i\in[N]}\max\{0,\,q_{\pi(i)}-q_i\}=\sum_{i\in[N]}\bigl(\max\{q_i,q_{\pi(i)}\}-q_i\bigr)\le N\eta.
\]
The left-hand side is the total increase of the prices along $\pi$, and since $\pi$ is a single cycle, all prices lie within $N\eta$ of each other (Lemma~\ref{lem:cycle}).
If each agent preferred only the piece in its own slot, the resulting inequalities $q_{\sigma(i)}\ge q_i-\eta$ would hold trivially for the identity assignment.

Hence, every piece indexed by $S_v$ is shorter by about $z_v$ than a piece in no set.
We arrange the sets so that the cut representing player $u$ is preceded by a fixed number of additional pieces indexed by $S_v$, compared with the average, exactly when $v$ is an in-neighbor of $u$.
This cut thus moves to the left in proportion to the total value of the in-neighbors of $u$ (Lemma~\ref{lem:decoding}), and reading off $z_u$ from its position gives an approximate equilibrium.
For chores, we use the costs $3|F|-U_i(F)$, under which every piece indexed by $S_v$ becomes longer by about $z_v/2$.
Reversing the arrangement of the sets and doubling the excess counts then gives the same shift of the cuts (Appendix~\ref{app:chores}).

\subsection{Related Work}
PPAD~\cite{Papadimitriou1994} is a complexity class of total search problems whose solutions are guaranteed to exist by a parity argument on directed graphs.
It contains computational versions of Sperner's lemma, which yields the existence of contiguous envy-free allocations~\cite{Su1999}, and of Brouwer's fixed-point theorem, and computing a Nash equilibrium is PPAD-complete~\cite{DaskalakisGP2009,ChenDT2009}.
Deng, Qi, and Saberi~\cite{DengQS2012} proved PPAD membership for approximate contiguous cake cutting and PPAD-hardness for general preferences.
Gao, Roghani, Rubinstein, and Saberi~\cite{GaoRRS2024} proved that, for general continuous valuations and any constant number $n\ge3$ of agents, it is PPAD-complete to find an allocation that makes three of the agents $\varepsilon$-envy-free, even if each agent may receive a constant number of pieces.

Hollender and Rubinstein~\cite{HollenderR2025} gave an efficient algorithm for four monotone agents and showed PPAD-hardness for four non-monotone agents, and Hollender, Maystre, and Risse~\cite{HollenderMR2026} showed intractability already for three non-monotone agents.
For five monotone agents, the complexity remains open~\cite{HollenderR2025,HollenderMR2026}.
These results concern a constant number of agents, whereas in our hard instances, the number of agents is part of the input.
This is necessary unless $\mathrm{PPAD}\subseteq\mathrm{FP}$, since for explicit piecewise-constant additive valuations, the exact problem is polynomial-time solvable for every fixed number of agents~\cite[Theorem~5]{Alijani+2017}.

Goldberg, Hollender, and Suksompong~\cite{GoldbergHS2020b} proved NP-hardness of several decision variants, including the one with a fixed agent order, and showed that the unary-precision problem with piecewise-constant valuations is equivalent to exact envy-freeness for binary valuations in which every good is valued by exactly one agent.
They also computed a contiguous $1/3$-envy-free allocation in polynomial time.
Barman and Kulkarni~\cite{BarmanK2023} improved the additive guarantee to $1/4+o(1)$ and gave an algorithm polynomial in $n$ and $1/\varepsilon$ when the number of distinct valuations is at most $\varepsilon n-1$.

Threshold games were introduced by Papadimitriou and Peng~\cite{PapadimitriouP2021} for public goods games in directed networks and were also used for the Hylland--Zeckhauser scheme~\cite{ChenCPY2022}.
Consensus halving, another fair-division problem with a topological existence proof, is PPA-complete~\cite{FilosRatsikasG2018,FilosRatsikasG2019}, even for constant precision~\cite{DeligkasFHM2025} and for piecewise-uniform valuations with two blocks per agent~\cite{FilosRatsikasHSZ2023}.

For goods on a path, Suksompong~\cite{Suksompong2019} proved that contiguous EF1 allocations exist for binary valuations.
Kawase, Roy, and Sanpui~\cite{KawaseRS2026} proved that deciding the existence of an EF1 allocation is NP-complete for a fixed agent order, whereas Theorem~\ref{thm:discrete} shows that computing one is PPAD-complete for a flexible order.

\section{Preliminaries}\label{sec:preliminaries}
We write $[k]=\{1,\ldots,k\}$, $|F|$ for the length of an interval $F$, and $|S|$ for the number of elements of a finite set $S$.
For a real number $t$, let $\clip(t)=\min\{1,\max\{0,t\}\}$.

\subsection{Contiguous cake cutting}\label{sec:cake}
An instance of our problem consists of $n$ agents and a positive rational $\varepsilon$.
Each agent $i$ has a normalized additive valuation $V_i$ on the cake $[0,1]$ given by a piecewise-constant density $\rho_i$ via
\[
 V_i([a,b])=\int_a^b\rho_i(t)\,dt,
\]
where $\rho_i$ is given by finitely many rational breakpoints in $[0,1]$ and nonnegative rational values on the intervals between them, all encoded in binary, and $V_i([0,1])=1$.
The precision $\varepsilon$ is encoded in binary unless we explicitly restrict it to unary encoding.
In the unary encoding, $\varepsilon=p/q$ for positive integers $p$ and $q$ written in unary, and hence $1/\varepsilon\le q$ is at most the input length.

A \emph{contiguous allocation} of an interval $[0,\Lambda]$ to $n$ agents consists of cuts $0=\theta_0\le\theta_1\le\cdots\le\theta_n=\Lambda$ and a bijection $\sigma\colon[n]\to[n]$.
The cuts divide $[0,\Lambda]$ into the \emph{pieces} $[\theta_{j-1},\theta_j]$ for $j\in[n]$, and agent $i$ receives the $\sigma(i)$th piece
\[
 X_i=[\theta_{\sigma(i)-1},\theta_{\sigma(i)}].
\]
Pieces may have length zero, and adjacent pieces share only an endpoint, which has zero value.
The allocation is \emph{$\varepsilon$-envy-free} if $V_i(X_i)\ge V_i(X_j)-\varepsilon$ for all $i,j$, and \emph{envy-free} if this holds for $\varepsilon=0$.
We use these definitions also for the unnormalized valuations on $[0,\Lambda]$ that appear in our reductions.

We call the problem of finding an $\varepsilon$-envy-free contiguous allocation of $[0,1]$ with rational cuts the \emph{\Cake{} problem}.\footnote{Its restriction to unary-encoded $\varepsilon$ is called \textsc{unary-$\varepsilon$-EF-Cake-Cutting} by Goldberg, Hollender, and Suksompong~\cite[Definition~6.3]{GoldbergHS2020b}.}
Every instance of this problem has a polynomial-size solution by Lemma~\ref{lem:membership} below, and solutions can be verified in polynomial time, since values at rational cuts are computable exactly.

A \emph{search reduction} maps each source instance to a target instance, and \emph{every} target solution back to a source solution, in polynomial time.
PPAD is a class of total search problems, that is, problems in which every instance has a polynomial-size solution that can be verified in polynomial time.
It consists of the problems that reduce to \textsc{End-of-Line}~\cite{Papadimitriou1994}.
Filos-Ratsikas, Hansen, H{\o}gh, and Hollender~\cite{FilosRatsikasHHH2026} showed that computing an envy-free contiguous allocation lies in PPAD for valuations given by piecewise-linear arithmetic circuits, which include the integrals of explicit piecewise-constant densities.
For $\varepsilon$-envy-freeness, membership also follows from earlier work~\cite{DengQS2012,HollenderR2025}.
Since an envy-free allocation is $\varepsilon$-envy-free for every $\varepsilon\ge0$, we obtain the following.

\begin{lemma}[\cite{DengQS2012,HollenderR2025,FilosRatsikasHHH2026}]\label{lem:membership}
For explicit rational piecewise-constant valuations, computing a contiguous envy-free allocation belongs to PPAD, and so does computing a contiguous $\varepsilon$-envy-free allocation, even when $\varepsilon$ is given in binary.
In particular, the \Cake{} problem belongs to PPAD.
\end{lemma}

In \emph{chore division}, the cake is a chore that every agent wants to avoid: the densities $\rho_i$ are nonpositive and $V_i([0,1])=-1$, and we write $C_i=-V_i$ for the cost of agent $i$.
The definitions above carry over, also to unnormalized valuations, and hence an allocation is $\varepsilon$-envy-free if $C_i(X_i)\le C_i(X_j)+\varepsilon$ for all $i,j$.
We consider chore division only in Section~\ref{sec:chores} and Appendix~\ref{app:chores}.

\subsection{Indivisible items on a path}\label{sec:items}
An instance of the discrete problem consists of items $g_1,\ldots,g_m$ on a path, in this order, and $n$ agents.
The items are either \emph{goods}, for which agent $i$ has an additive valuation $\nu_i$ that assigns each good a nonnegative rational value, or \emph{chores}, for which agent $i$ has an additive cost $\kappa_i$ that assigns each chore a nonnegative rational cost.
The valuations or costs are \emph{binary} if every item has value or cost $0$ or $1$.
A \emph{contiguous allocation} is given by indices and a bijection $\sigma\colon[n]\to[n]$,
\[
 0=\beta_0\le\beta_1\le\cdots\le\beta_n=m,\qquad
 X_i=\{g_k:\beta_{\sigma(i)-1}<k\le\beta_{\sigma(i)}\},
\]
where $X_i$ is the bundle of agent $i$.
For an integer $c\ge0$, an allocation of goods is \emph{EF$c$} if for all $i,j\in[n]$ there is a set $G\subseteq X_j$ with $|G|\le c$ and $\nu_i(X_i)\ge\nu_i(X_j\setminus G)$, and an allocation of chores is \emph{EF$c$} if for all $i,j\in[n]$ there is a set $G\subseteq X_i$ with $|G|\le c$ and $\kappa_i(X_i\setminus G)\le\kappa_i(X_j)$.
In both cases, the allocation is \emph{EF$c_{\mathrm{outer}}$} if, in addition, $G$ can be chosen so that $X_j\setminus G$ for goods and $X_i\setminus G$ for chores is connected or empty~\cite{Bilo+2019,Igarashi2023}.

\subsection{Threshold games}\label{sec:source}
We reduce from threshold games, which were introduced by Papadimitriou and Peng~\cite{PapadimitriouP2021}.
A threshold game is given by a directed graph with vertex set $[s]$, whose vertices are players.
Each player $v$ chooses a strategy $z_v\in[0,1]$.
Let $A\in\{0,1\}^{s\times s}$ be the matrix with $A_{vu}=1$ if and only if $(u,v)$ is an arc.
Then, $(Az)_v$ is the sum of the strategies of the in-neighbors of $v$, and player $v$ prefers strategy $1$ if $(Az)_v<1/2$ and strategy $0$ if $(Az)_v>1/2$.
For $\delta\ge0$, a vector $z\in[0,1]^s$ is a \emph{$\delta$-approximate equilibrium} if, for every $v\in[s]$,
\[
 (Az)_v<\tfrac12-\delta\implies z_v\ge1-\delta,
 \qquad
 (Az)_v>\tfrac12+\delta\implies z_v\le\delta.
\]
We call the problem of finding a $\delta$-approximate equilibrium the \emph{$\delta$-\TG{} problem}.
Deligkas, Fearnley, Hollender, and Melissourgos~\cite{DeligkasFHM2024} showed that this problem is PPAD-complete below the threshold $1/6$ and solvable in polynomial time at $1/6$.

\begin{theorem}[{\cite[Theorem~4.23]{DeligkasFHM2024}}]\label{thm:source}
The $\delta$-\TG{} problem is PPAD-complete for every constant $0\le\delta<1/6$, even when every vertex has indegree and outdegree at most two.
\end{theorem}

\section{PPAD-hardness for additive cake valuations}\label{sec:hardness}
In this section, we prove Theorem~\ref{thm:main} by a search reduction from the $1/12$-\TG{} problem with indegrees at most two, which is PPAD-hard by Theorem~\ref{thm:source}.
Fix such a game with vertex set $[s]$, and let $A\in\{0,1\}^{s\times s}$ be its matrix as in Section~\ref{sec:source}.
We use only the bound on the indegrees: every row of $A$ has at most two ones.
The construction also allows self-loops, that is, $A_{vv}=1$.

We represent the strategy $z_v$ of each player $v$ by the position of one cut.
Ignoring approximation errors, our construction enforces
\[
 z_v=\clip\bigl(\tfrac{13}{2}-12(Az)_v\bigr)\qquad(v\in[s]),
\]
which is $1$ if $(Az)_v\le11/24$ and $0$ if $(Az)_v\ge13/24$, and hence $z$ is an approximate equilibrium.
To this end, the cut of player $v$ lies about $12(Az)_v$ to the left of a reference point, and $z_v$ is the length of the part to the left of the cut of a unit interval centered $6$ to the left of this point.
As $(Az)_v$ increases from $11/24$ to $13/24$, the cut moves across this interval from its right end to its left end, and hence $z_v$ decreases from $1$ to $0$.
The factor $12$ stretches this window of width $1/12$ to the length $1$ of the interval, and the window lies inside $[5/12,7/12]$ with a margin of $1/24$ on each side, which absorbs the approximation errors.
Section~\ref{sec:construction} gives the construction, and Section~\ref{sec:soundness} shows how to construct an approximate equilibrium from every approximately envy-free allocation and completes the proof.

\subsection{The construction}\label{sec:construction}
From the game, we construct a cake-cutting instance with $N=32(s+1)^2$ agents, identified with $[N]$, and additive valuations $U_1,\ldots,U_N$, which are not normalized, on the cake $[0,L]$, where $L=40N^2$.
The cake consists of $N$ consecutive \emph{slots} of length $B=40N$, where the $j$th slot is $[(j-1)B,jB]$, and hence $L=NB$.
We use the same index set for agents, slots, and pieces.
We first define the sets $S_v$ and a cyclic order on the agents, then the valuations, and finally the target instance of the \Cake{} problem, which is obtained by normalizing the valuations.
We use two orders on the indices: the spatial order, in which the numbers of indices of each set $S_v$ in prefixes of $[N]$ encode $A$, and a cyclic order, which visits the indices of each $S_v$ consecutively and handles arbitrary assignments in Section~\ref{sec:soundness}.

\paragraph{The sets $S_v$ and the cyclic order.}
We construct pairwise disjoint sets $S_1,\ldots,S_s\subseteq[N]$, each of size $R=32(s+1)$, and hence $N=(s+1)R$.
The cut of player $v$ is the cut after the $a_v$th piece, where $a_v=vR$.
The sets encode $A$: among the first $a_u$ indices, $S_v$ has twelve more indices than the average $32u$ exactly when $v$ is an in-neighbor of $u$.
In Section~\ref{sec:soundness}, we show that each piece indexed by $S_v$ is shorter by approximately $z_v$ than a piece in no set, and hence these twelve extra pieces move the cut of player $u$ to the left by approximately $12z_v$.
We construct the sets as follows.
Partition $[N]$ into $s+1$ consecutive bins of $R$ indices, where the $u$th bin is $\{(u-1)R+1,\ldots,uR\}$, and let $A_{0v}=A_{s+1,v}=0$.
Start with a block of $32$ indices for each $S_v$ in each bin, and for every arc $(v,u)$, transfer twelve indices of $S_v$ from the $(u+1)$st bin to the $u$th bin.
For each $u\in[s+1]$, we fill the right end of the $u$th bin with consecutive blocks, one for each $v\in[s]$, where the block for $S_v$ consists of $32+12A_{uv}-12A_{u-1,v}$ indices.
For $u\in[s]$, the block for $S_u$ comes last, and the other blocks appear in increasing order of $v$.
In the last bin, all blocks appear in increasing order of $v$.
The remaining indices at the left end of the bin lie in no set.

By the formula above, every block has at least $32-12=20$ indices, and since every indegree is at most two, the blocks in each bin have at most $32s+24=R-8$ indices in total.
Hence, for $u\in[s]$, the last index $a_u$ of the $u$th bin lies in $S_u$, whose block is nonempty and comes last, and for $u\in[s+1]$, the first index $a_{u-1}+1$ of the $u$th bin lies in no set, where $a_0=0$.
Figure~\ref{fig:example} illustrates the sets on a small example.

We also use a cyclic order of the indices.
Consider the list of all indices
\[
 a_1+1,\ S_1\setminus\{a_1\},\ a_1,\ a_2+1,\ S_2\setminus\{a_2\},\ a_2,\ \ldots,\ a_s+1,\ S_s\setminus\{a_s\},\ a_s,\ \text{the remaining indices},
\]
where the indices of each $S_v\setminus\{a_v\}$ are listed in a fixed order, for example as in Figure~\ref{fig:example}, and the remaining indices are those not listed before, in increasing order.
Let $\pi\colon[N]\to[N]$ map each index to the next index in this list, and the last index to $a_1+1$.
Thus, $\pi$ is a permutation consisting of a single cycle.
For an agent $i$, we call the $\pi(i)$th slot its \emph{next slot}.
Here, $a_v+1$ is the index immediately after $a_v$ in spatial order, whereas $\pi(i)$ is the successor of $i$ in cyclic order.
Both the sets and $\pi$ can be computed in polynomial time.

The sets and the cyclic order have the following properties for all $u,v\in[s]$:
\begin{enumerate}[label=\textup{(S\arabic*)},leftmargin=*]
 \item\label{S:counts} $|S_v|=R$ and $|S_v\cap[a_u]|=32u+12A_{uv}$;
 \item\label{S:arc} along $\pi$, the indices of $S_v$ are visited consecutively and end at $a_v$, the index $a_v+1$ visited just before them and the index $\pi(a_v)$ visited just after them lie in no set, and $\pi(a_v)\ne a_v+1$ and $\pi(a_v+1)\ne a_v$.
\end{enumerate}
Property~\ref{S:counts} follows by summing the block sizes over all bins and over the first $u$ bins, respectively.
Property~\ref{S:arc} follows from the definition of $\pi$: the index $\pi(a_v)$ is $a_{v+1}+1$ if $v<s$ and the first remaining index if $v=s$, which exists since the index $1$ lies in no set and differs from every $a_w+1\ge R+1$, and $\pi(a_v+1)$ is an index of $S_v$ other than $a_v$, which exists since $|S_v|=R\ge2$.

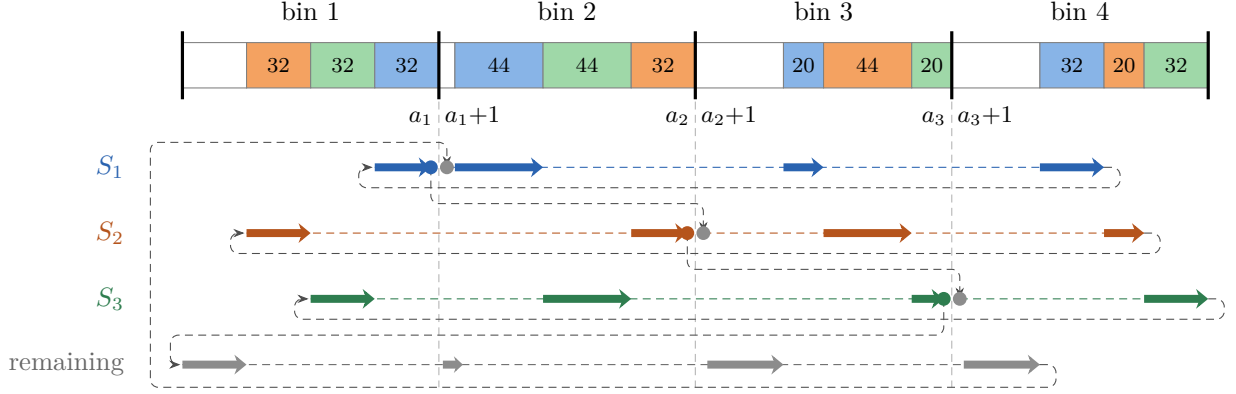
\begin{figure}[t]
\centering
\begin{tikzpicture}[x=0.0265cm,y=0.6cm,font=\small]
 \def\pos{0}
 \foreach \c/\w in {0/32,2/32,3/32,1/32,0/8,1/44,3/44,2/32,0/44,1/20,2/44,3/20,0/44,1/32,2/20,3/32} {
  \ifcase\c \def\col{white}\or \def\col{setone}\or \def\col{settwo}\or \def\col{setthree}\fi
  \pgfmathsetmacro\np{\pos+\w}
  \pgfmathsetmacro\mid{(\pos+\np)/2}
  \filldraw[fill=\col,draw=black!50] (\pos,0) rectangle (\np,1);
  \ifnum\c>0 \node[font=\scriptsize] at (\mid,0.5) {\w};\fi
  \xdef\pos{\np}
 }
 \foreach \b in {0,128,256,384,512} \draw[very thick] (\b,-0.25) -- (\b,1.25);
 \foreach \u in {1,2,3,4} \node at (128*\u-64,1.7) {bin \u};
 \tikzset{
  pimove/.style={black!70,densely dashed,-{Stealth[length=1.5mm]},rounded corners=3pt},
  runline/.style={densely dashed},
  run/.style={line width=3pt,-{Stealth[length=1.6mm,width=2.6mm,inset=0.5mm]}}}
 \def\ra{-1.75}\def\rb{-3.2}\def\rc{-4.65}\def\rd{-6.1}
 \foreach \u in {1,2,3} {
  \draw[black!30,densely dashed] (128*\u,-0.3) -- (128*\u,\rd-0.5);
  \node[anchor=base east,font=\footnotesize,inner sep=2pt] at (128*\u,-0.75) {$a_{\u}$};
  \node[anchor=base west,font=\footnotesize,inner sep=2pt] at (128*\u,-0.75) {$a_{\u}{+}1$};
 }
 \foreach \y/\t/\col in {\ra/{$S_1$}/setonedark,\rb/{$S_2$}/settwodark,\rc/{$S_3$}/setthreedark,\rd/{remaining}/black!60}
  \node[left,text=\col] at (-24,\y) {\t};
 \draw[setonedark,runline] (132,\ra) -- (460,\ra);
 \foreach \a/\b in {136/180,300/320,428/460} \draw[setonedark,run] (\a,\ra) -- (\b,\ra);
 \draw[pimove] (460,\ra) -- (468,\ra) -- (468,\ra-0.45) -- (88,\ra-0.45) -- (88,\ra) -- (95,\ra);
 \draw[setonedark,runline] (96,\ra) -- (124,\ra);
 \draw[setonedark,run] (96,\ra) -- (124,\ra);
 \draw[pimove] (124,\ra) -- (124,\ra-0.8) -- (260,\ra-0.8) -- (260,\rb);
 \draw[settwodark,runline] (260,\rb) -- (480,\rb);
 \foreach \a/\b in {320/364,460/480} \draw[settwodark,run] (\a,\rb) -- (\b,\rb);
 \draw[pimove] (480,\rb) -- (488,\rb) -- (488,\rb-0.45) -- (24,\rb-0.45) -- (24,\rb) -- (31,\rb);
 \draw[settwodark,runline] (32,\rb) -- (252,\rb);
 \foreach \a/\b in {32/64,224/252} \draw[settwodark,run] (\a,\rb) -- (\b,\rb);
 \draw[pimove] (252,\rb) -- (252,\rb-0.8) -- (388,\rb-0.8) -- (388,\rc);
 \draw[setthreedark,runline] (388,\rc) -- (512,\rc);
 \draw[setthreedark,run] (480,\rc) -- (512,\rc);
 \draw[pimove] (512,\rc) -- (520,\rc) -- (520,\rc-0.45) -- (56,\rc-0.45) -- (56,\rc) -- (63,\rc);
 \draw[setthreedark,runline] (64,\rc) -- (380,\rc);
 \foreach \a/\b in {64/96,180/224,364/380} \draw[setthreedark,run] (\a,\rc) -- (\b,\rc);
 \draw[pimove] (380,\rc) -- (380,\rc-0.8) -- (-6,\rc-0.8) -- (-6,\rd) -- (-1,\rd);
 \draw[black!60,runline] (0,\rd) -- (428,\rd);
 \foreach \a/\b in {0/32,262/300,390/428} \draw[black!45,run] (\a,\rd) -- (\b,\rd);
 \draw[black!45,run,-{Stealth[length=1.2mm,width=2.2mm,inset=0.4mm]}] (130,\rd) -- (140,\rd);
 \draw[pimove] (428,\rd) -- (436,\rd) -- (436,\rd-0.5) -- (-16,\rd-0.5) -- (-16,\ra+0.55) -- (132,\ra+0.55) -- (132,\ra);
 \foreach \x/\y/\col in {124/\ra/setonedark,132/\ra/black!45,252/\rb/settwodark,260/\rb/black!45,380/\rc/setthreedark,388/\rc/black!45} \fill[\col] (\x,\y) circle (2.6pt);
\end{tikzpicture}
\caption{The sets $S_v$ for the game with arcs $(1,2)$, $(2,3)$, and $(3,2)$, where $s=3$, $R=128$, and $N=512$. Here, $A_{21}=A_{23}=A_{32}=1$, and all other entries of $A$ are $0$. Blue, orange, and green blocks consist of the indices of $S_1$, $S_2$, and $S_3$, respectively, and are labeled with their sizes. White blocks consist of indices in no set. The rows below trace the cyclic order~$\pi$, starting from $a_1+1$: thick arrows run through the indices of a set, and dashed arrows lead to the next index along~$\pi$.}
\label{fig:example}
\end{figure}

\paragraph{Valuations.}
The valuation $U_i$ of each agent $i$ has a density $f_i$, which is $1$ everywhere, plus bonuses of height $1$ on at most three short intervals.
The \emph{core} of the $j$th slot and the \emph{boundary interval} of player $v$ are
\[
 I_j=[(j-1)B+10N,\ jB-10N]\quad(j\in[N]),\qquad
 J_v=[a_vB-\tfrac{13}{2},\ a_vB-\tfrac{11}{2}]\quad(v\in[s]).
\]
Each core lies in the middle of its slot and has length $20N$.
The boundary interval $J_v$ has length $1$ and lies in the $a_v$th slot, near its right end $a_vB$ (Figure~\ref{fig:boundary}).

By~\ref{S:arc}, the cycle $\pi$ enters $S_v$ just after agent $a_v+1$ and leaves it at agent $a_v$, the two agents whose slots lie on either side of the cut of player $v$.
On top of density $1$, agent $i$ has additional density $1$ on an interval centered in the core $I_i$ of its own slot, on an interval centered in the core $I_{\pi(i)}$ of its next slot, and, if $i\in\{a_v,a_v+1\}$, on the boundary interval $J_v$.
These \emph{bonuses} have the following lengths:
\begin{center}
\setlength{\tabcolsep}{12pt}
\begin{tabular}{lccc}
 \toprule
 agent $i$ & in $I_i$ & in $I_{\pi(i)}$ & on $J_v$\\
 \midrule
 $a_v+1$ & $2$ & $3$ & $1$\\
 $a_v$ & $5/2$ & $5/2$ & $1$\\
 otherwise & $3$ & $3$ & $0$\\
 \bottomrule
\end{tabular}
\end{center}
Thus, every agent slightly prefers its own slot and its next slot, the two agents next to the cut of player $v$ also value $J_v$, and every agent has total bonus $6$.
These lengths are chosen for Lemma~\ref{lem:envelope}.
If a cut splits $J_v$ into parts of lengths $z_v$ and $1-z_v$ as in Figure~\ref{fig:boundary}, then the bonus of agent $a_v+1$ on the piece in its next slot exceeds that on the piece in its own slot by $z_v$, and the bonus of agent $a_v$ on the piece in its next slot falls short of that on the piece in its own slot by $z_v$.
Since the pairs $\{a_v,a_v+1\}$ are pairwise disjoint, each agent has at most one boundary bonus.

\begin{figure}[t]
\centering
\begin{tikzpicture}[x=1cm,y=1cm,font=\small]
 \fill[black!7] (-0.6,0) -- (-0.6,1.2) -- (1.875,1.2) -- (1.875,2.2) -- (4.125,2.2) -- (4.125,1.2)
  -- (4.85,1.2) -- (4.85,2.2) -- (5.75,2.2) -- (5.75,1.2) -- (12.6,1.2) -- (12.6,0) -- cycle;
 \foreach \x in {1.875,4.125,4.85,5.75} \draw[gray!60,dotted] (\x,0.12) -- (\x,1.2);
 \draw[dotted] (-0.6,1.2) -- (12.6,1.2);
 \draw[dotted] (-0.6,2.2) -- (12.6,2.2);
 \node[left] at (-0.6,1.2) {$1$};
 \node[left] at (-0.6,2.2) {$2$};
 \draw[thick] (-0.6,1.2) -- (1.875,1.2) -- (1.875,2.2) -- (4.125,2.2) -- (4.125,1.2)
  -- (4.85,1.2) -- (4.85,2.2) -- (5.75,2.2) -- (5.75,1.2) -- (12.6,1.2);
 \node[right] at (12.6,1.2) {$f_{a_v}$};
 \draw[thick] (-0.6,0) -- (12.6,0);
 \draw[thick,dotted] (-1.1,0) -- (-0.6,0);
 \draw[thick,dotted] (12.6,0) -- (13.1,0);
 \foreach \x/\l in {0/{(a_v-1)B},6/{a_vB},12/{(a_v+1)B}}{
  \draw[thick] (\x,-0.16) -- (\x,0.16);
  \node[below] at (\x,-0.18) {$\l$};
 }
 \draw[coredark,line width=2pt,|-|] (1.5,0) -- (4.5,0);
 \draw[coredark,line width=2pt,|-|] (7.5,0) -- (10.5,0);
 \draw[bnddark,line width=2pt,|-|] (4.85,0) -- (5.75,0);
 \node[above,text=coredark] at (3,0.1) {$I_{a_v}$};
 \node[above,text=coredark] at (9,0.1) {$I_{a_v+1}$};
 \node[above right,text=bnddark] at (5.27,0.18) {$J_v$};
 \draw[dashed,thick] (5.25,-0.85) -- (5.25,3.05);
 \node[above] at (5.25,3.05) {$\theta_{a_v}$};
 \draw[decorate,decoration={brace,amplitude=3pt}] (4.85,2.3) -- (5.25,2.3)
  node[midway,above=3pt,anchor=south east,xshift=-1pt] {$z_v$};
 \draw[decorate,decoration={brace,amplitude=3pt}] (5.25,2.3) -- (5.75,2.3)
  node[midway,above=3pt,anchor=south west,xshift=1pt] {$1-z_v$};
 \draw[|<->|] (0.3,-0.85) -- (5.25,-0.85) node[midway,below] {$P_{a_v}$};
 \draw[|<->|] (5.25,-0.85) -- (11.7,-0.85) node[midway,below] {$P_{a_v+1}$};
\end{tikzpicture}
\caption{The cake around the cut $\theta_{a_v}$ (not to scale). The interval $J_v$ is valued by the agents $a_v$ and $a_v+1$. Above the cake, we draw the density $f_{a_v}$ of agent $a_v$, whose next core lies elsewhere. It has bonuses in the middle of its own core $I_{a_v}$ and on $J_v$. The lengths $z_v$ and $1-z_v$ of the two parts of $J_v$ and the pieces $P_{a_v}$ and $P_{a_v+1}$ below the cake are used in Section~\ref{sec:soundness}. The figure shows a cut inside $J_v$.}
\label{fig:boundary}
\end{figure}
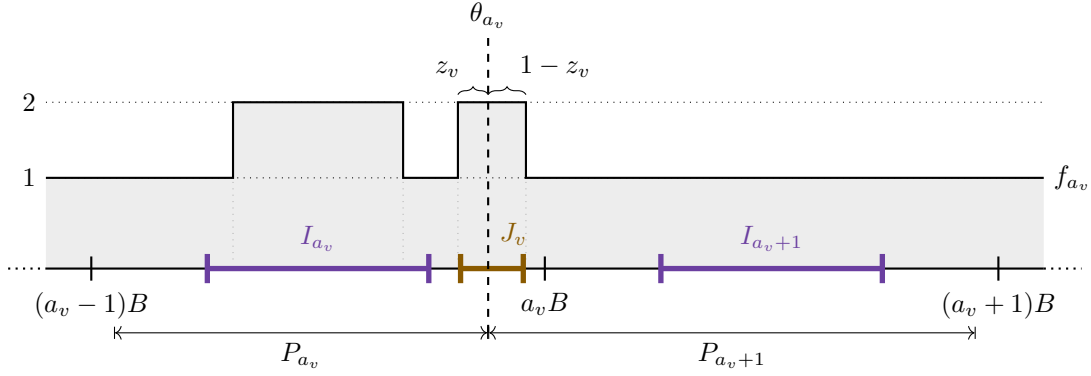

\begin{lemma}\label{lem:values}
Every $f_i$ takes only the values $1$ and $2$ and is piecewise constant with at most six breakpoints.
Every interval $F$ satisfies
\begin{equation}\label{eq:uniform}
 0\le U_i(F)-|F|\le6,\qquad\text{and}\qquad U_i([0,L])=L+6.
\end{equation}
\end{lemma}
\begin{proof}
The bonuses of agent $i$ lie in $I_i$, in $I_{\pi(i)}$, and possibly on one boundary interval.
The first two have length at most $3$ and hence fit in the cores, and the three intervals are disjoint, since $\pi(i)\ne i$ and every boundary interval lies outside all cores.
Hence, $f_i$ takes only the values $1$ and $2$, and its only breakpoints are the at most six endpoints of these intervals.
The total bonus is $6$, which gives~\eqref{eq:uniform}.
\end{proof}

\paragraph{The target instance.}
Since $U_i([0,L])=L+6$ by Lemma~\ref{lem:values}, the target instance of the \Cake{} problem consists of the normalized valuations and the unary-encoded precision
\[
 V_i([a,b])=\frac{U_i([La,Lb])}{L+6}\quad(i\in[N]),\qquad
 \varepsilon=\frac1{8000N^4}.
\]
The density of $V_i$ is $\rho_i(t)=Lf_i(Lt)/(L+6)$.
By Lemma~\ref{lem:values}, it takes only the values $\lambda$ and $2\lambda$ for $\lambda=L/(L+6)$, which is the same for all agents, and it has at most six breakpoints.

Let $(X_1,\ldots,X_N)$ be any $\varepsilon$-envy-free contiguous allocation for $V_1,\ldots,V_N$, where $X_i$ is the piece of agent $i$.
Multiplying its cuts by $L$ gives the allocation $(LX_1,\ldots,LX_N)$ of $[0,L]$, where $LX=\{Lt:t\in X\}$.
Since $L+6\le2L$ and $L=40N^2$, it satisfies
\[
 U_i(LX_j)-U_i(LX_i)=(L+6)\bigl(V_i(X_j)-V_i(X_i)\bigr)\le2L\varepsilon=\frac1{100N^2}\qquad(i,j\in[N]).
\]
Hence, it is $\eta$-envy-free for $U_1,\ldots,U_N$, where $\eta=1/(100N^2)$.

\subsection{From allocations to equilibria}\label{sec:soundness}
We now show how to construct a solution of the source instance from any solution of the target instance.
By Section~\ref{sec:construction}, it suffices to construct a $1/12$-approximate equilibrium from every $\eta$-envy-free allocation of $[0,L]$ for $U_1,\ldots,U_N$, where $\eta=1/(100N^2)$.
Throughout this subsection, we fix such an allocation, given by cuts $0=\theta_0\le\theta_1\le\cdots\le\theta_N=L$ and a bijection $\sigma\colon[N]\to[N]$.
Let $P_j=[\theta_{j-1},\theta_j]$ be the $j$th piece from the left, so that agent $i$ receives $P_{\sigma(i)}$.
Let $\ell_j=|P_j|$ be its length.
Then, $\eta$-envy-freeness means that
\[
 U_i(P_{\sigma(i)})\ge U_i(P_j)-\eta\qquad(i,j\in[N]).
\]
We show that the vector $z=(z_1,\ldots,z_s)$ given by
\begin{equation}\label{eq:decoder}
 z_v=\clip\bigl(\theta_{a_v}-a_vB+\tfrac{13}{2}\bigr)\qquad(v\in[s])
\end{equation}
is a $1/12$-approximate equilibrium.
We proceed in three steps.
First, for every $j$, the $j$th piece lies close to the $j$th slot (localization).
Second, regardless of the assignment $\sigma$, the length of every piece is determined by $z$ up to a small error (prices).
Third, the cut of each player $v$ lies at a position determined by the strategies of its in-neighbors, and hence $z$ is an approximate equilibrium (decoding).

\paragraph{Localization.}
Every agent values each interval at its length plus at most $6$.
Hence, approximate envy-freeness forces all piece lengths to be close to $B$ and each cut $\theta_j$ to be close to $jB$.

\begin{lemma}\label{lem:localization}
Every cut satisfies $|\theta_j-jB|<7N$. Consequently:
\begin{enumerate}[label=\textup{(\alph*)},leftmargin=*]
 \item\label{L:core} $I_j\subseteq P_j$ and $P_j\cap I_k=\emptyset$ for every $k\ne j$;
 \item\label{L:boundary} $J_v\subseteq P_{a_v}\cup P_{a_v+1}$ and $|J_v\cap P_{a_v}|=z_v$ for every $v$.
\end{enumerate}
\end{lemma}
\begin{proof}
Let agent $i$ receive a shortest piece, and let $P_{j'}$ be a longest piece.
By~\eqref{eq:uniform} and $\eta$-envy-freeness,
\[
 \ell_{j'}\le U_i(P_{j'})\le U_i(P_{\sigma(i)})+\eta\le\ell_{\sigma(i)}+6+\eta.
\]
As the average length is $B$, every $\ell_j$ lies in $(B-7,B+7)$.
Since $\theta_j-jB=\sum_{k\in[j]}(\ell_k-B)$, we obtain $|\theta_j-jB|<7N$.
Thus, $\theta_{j-1}<(j-1)B+10N$ and $\theta_j>jB-10N$, which gives $I_j\subseteq P_j$.
Similarly, $P_j$ meets no neighboring core, and hence no other core.
Since $B-7N>\frac{13}{2}$, we also have $\theta_{a_v-1}<a_vB-\frac{13}{2}$ and $\theta_{a_v+1}>a_vB$.
Thus, $J_v=[a_vB-\frac{13}{2},a_vB-\frac{11}{2}]$ lies in $P_{a_v}\cup P_{a_v+1}$, and the part in $P_{a_v}$ has length $\clip(\theta_{a_v}-a_vB+\frac{13}{2})=z_v$.
Figure~\ref{fig:boundary} shows these two parts of $J_v$.
\end{proof}

\paragraph{Prices.}
Localization determines where each piece lies, but not which agent receives it, and it determines the lengths only up to an error of $7$.
A bonus in its own slot alone would give each agent only one-sided inequalities between lengths.
We therefore make each agent $i$ indifferent, at suitable prices, between the two pieces $P_i$ and $P_{\pi(i)}$, so that approximate envy-freeness forces the piece received by agent $i$, whichever it is, to cost almost as much as both.
Summing over agents removes the assignment permutation, and the single cycle connects all prices.
The following combinatorial lemma makes this precise.

\begin{lemma}\label{lem:cycle}
Let $\pi$ and $\sigma$ be permutations of $[N]$ such that $\pi$ consists of a single cycle, and let $\eta\ge0$.
If reals $q_1,\ldots,q_N$ satisfy $q_{\sigma(i)}\ge\max\{q_i,q_{\pi(i)}\}-\eta$ for every $i\in[N]$, then $\max_jq_j-\min_jq_j\le N\eta$.
\end{lemma}
\begin{proof}
Summing the assumption over $i\in[N]$ and using $\sum_{i\in[N]}q_{\sigma(i)}=\sum_{i\in[N]}q_i$ gives
\[
 \sum_{i\in[N]}\max\{0,\,q_{\pi(i)}-q_i\}
 =\sum_{i\in[N]}\bigl(\max\{q_i,q_{\pi(i)}\}-q_i\bigr)\le N\eta.
\]
The $i$th term on the left is the increase of $q$ at the step from $i$ to $\pi(i)$.
Since $\pi$ is a single cycle, it leads from an index minimizing $q$ to an index maximizing $q$.
The increases along this path sum to at least $\max_jq_j-\min_jq_j$, which is therefore at most $N\eta$.
\end{proof}

In our setting, the assumption of Lemma~\ref{lem:cycle} follows from a bound on value differences that is tight for the next piece.
Suppose that reals $q_j$ satisfy
\begin{equation}\label{eq:envelope}
\begin{aligned}
 U_i(P_j)-U_i(P_i)&\le q_j-q_i&&\text{for all }i,j\in[N],\\
 U_i(P_{\pi(i)})-U_i(P_i)&=q_{\pi(i)}-q_i&&\text{for all }i\in[N].
\end{aligned}
\end{equation}
In other words, every agent $i$ with quasi-linear utility $U_i(P_j)-q_j$ weakly prefers both $P_i$ and $P_{\pi(i)}$ to every other piece, and we call $q_j$ the \emph{price} of $P_j$, although no money is involved.
By~\eqref{eq:envelope} for $j=\sigma(i)$, $\eta$-envy-freeness towards $P_i$ and $P_{\pi(i)}$, and the equality case of~\eqref{eq:envelope}, we have
\[
 q_{\sigma(i)}-q_i\ge U_i(P_{\sigma(i)})-U_i(P_i)\ge\max\{U_i(P_i),U_i(P_{\pi(i)})\}-\eta-U_i(P_i)=\max\{0,\,q_{\pi(i)}-q_i\}-\eta.
\]
Hence, Lemma~\ref{lem:cycle} gives $\max_jq_j-\min_jq_j\le N\eta$.

It remains to choose the prices.
Since agent $i$ values $P_j$ at its length $\ell_j$ plus bonuses, we take as the price of $P_j$ its length plus the decoded value of its set,
\[
 q_j=\begin{cases}
  \ell_j+z_v&\text{if }j\in S_v,\\
  \ell_j&\text{if $j$ is in no set.}
 \end{cases}
\]
The term $z_v$ absorbs the boundary bonuses, which differ between the pieces in the own slot and the next slot of an agent.
The next lemma shows that these prices satisfy~\eqref{eq:envelope}, and hence they are all nearly equal.

\begin{lemma}\label{lem:envelope}
The pieces $P_1,\ldots,P_N$ satisfy~\eqref{eq:envelope} with these prices.
\end{lemma}
\begin{proof}
By Lemma~\ref{lem:localization}\ref{L:core}, the piece $P_j$ contains $I_j$ and is disjoint from every other core.
Thus, agent $i$ values $P_j$ at $\ell_j$ plus a core bonus and a boundary bonus, and the core bonus is $0$ unless $j\in\{i,\pi(i)\}$.
If $i\in\{a_v,a_v+1\}$, then its boundary bonus on $P_j$ is $|J_v\cap P_j|$, which is $z_v$ for $j=a_v$, $1-z_v$ for $j=a_v+1$, and $0$ otherwise by Lemma~\ref{lem:localization}\ref{L:boundary}.
In particular, the boundary bonus is at most $1$, and it is $0$ on $P_{\pi(i)}$ by~\ref{S:arc}.
Also by~\ref{S:arc}, the cycle never moves directly from one set $S_v$ to another, and hence $q_j-\ell_j$ changes along $\pi$ only at the step from $a_v+1$, from $0$ to $z_v$, and at the step from $a_v$, from $z_v$ to $0$.
Therefore, the bonuses of agent $i$ on $P_i$ and $P_{\pi(i)}$ are as follows:
\begin{center}
\setlength{\tabcolsep}{12pt}
\begin{tabular}{lccc}
 \toprule
 agent $i$ & bonus on $P_i$ & bonus on $P_{\pi(i)}$ & $(q_{\pi(i)}-\ell_{\pi(i)})-(q_i-\ell_i)$\\
 \midrule
 $a_v+1$ & $2+(1-z_v)$ & $3$ & $z_v$\\
 $a_v$ & $5/2+z_v$ & $5/2$ & $-z_v$\\
 otherwise & $3$ & $3$ & $0$\\
 \bottomrule
\end{tabular}
\end{center}
In each case, the bonus on $P_{\pi(i)}$ exceeds that on $P_i$ by the last column, and hence
\[
 U_i(P_{\pi(i)})-U_i(P_i)=(\ell_{\pi(i)}-\ell_i)+\bigl((q_{\pi(i)}-\ell_{\pi(i)})-(q_i-\ell_i)\bigr)=q_{\pi(i)}-q_i.
\]
Here, the bonus of agent $a_v+1$ in its own core has length $2$ rather than $3$, which compensates for the constant $1$ in its boundary bonus $1-z_v$.
For every $j\notin\{i,\pi(i)\}$, the bonus of agent $i$ on $P_j$ is at most $1$, whereas its bonus on $P_i$ is at least $2$.
Since $0\le q_k-\ell_k\le1$ for every $k$, we obtain, for every $j\notin\{i,\pi(i)\}$,
\[
 U_i(P_j)-U_i(P_i)\le(\ell_j+1)-(\ell_i+2)=\ell_j-\ell_i-1\le q_j-(q_i-1)-1=q_j-q_i.
\]
Together with the equality for $j=\pi(i)$ and the trivial case $j=i$, this proves~\eqref{eq:envelope}.
\end{proof}

\paragraph{Decoding.}
By the previous step, every piece associated with player $v$ is shorter by about $z_v$ than a piece in no set.
Property~\ref{S:counts} places twelve extra such pieces before the cut of player $u$ for each arc $(v,u)$, shifting the cut to the left by about $12(Az)_u$.

\begin{lemma}\label{lem:decoding}
For every $u\in[s]$, we have $|\theta_{a_u}-a_uB+12(Az)_u|\le\frac1{100}$.
\end{lemma}
\begin{proof}
Let $\bar q=\frac1N\sum_{j\in[N]}q_j$ be the average price.
By Lemmas~\ref{lem:cycle} and~\ref{lem:envelope}, we have $\max_jq_j-\min_jq_j\le N\eta$, and hence $\bigl|\sum_{j\in[a_u]}(q_j-\bar q)\bigr|\le a_uN\eta\le N^2\eta=\frac1{100}$.
By the definition of the prices, $q_j-\ell_j$ is $z_v$ for $j\in S_v$ and $0$ for $j$ in no set.
Hence, $\sum_{j\in[N]}(q_j-\ell_j)=R\sum_{v\in[s]}z_v$ by~\ref{S:counts}, and since $\sum_{j\in[N]}\ell_j=NB$, we have $\bar q=B+\frac RN\sum_{v\in[s]}z_v$.
Similarly, $\sum_{j\in[a_u]}(q_j-\ell_j)=\sum_{v\in[s]}|S_v\cap[a_u]|\,z_v$, and hence property~\ref{S:counts} and $a_uR/N=32u$ give
\begin{align*}
 \theta_{a_u}-a_uB
 &=\sum_{j\in[a_u]}(\ell_j-B)
 =\sum_{j\in[a_u]}(q_j-\bar q)+a_u(\bar q-B)-\sum_{j\in[a_u]}(q_j-\ell_j)\\
 &=\sum_{j\in[a_u]}(q_j-\bar q)-\sum_{v\in[s]}\Bigl(|S_v\cap[a_u]|-\frac{a_uR}N\Bigr)z_v
 =\sum_{j\in[a_u]}(q_j-\bar q)-12(Az)_u.
\end{align*}
Since the first sum has absolute value at most $\frac1{100}$, the lemma follows.
\end{proof}

We are now ready to prove Theorem~\ref{thm:main}.

\begin{proof}[Proof of Theorem~\ref{thm:main}]
Membership is Lemma~\ref{lem:membership}.
For hardness, consider the search reduction that maps a game to the target instance of Section~\ref{sec:construction}.
Since $N=32(s+1)^2$ and $L=40N^2$, the target instance can be computed in polynomial time, and the unary encoding of $\varepsilon$ has length $O(s^8)$.
Given any $\varepsilon$-envy-free allocation for $V_1,\ldots,V_N$, multiplying its cuts by $L$ gives an $\eta$-envy-free allocation for $U_1,\ldots,U_N$ by Section~\ref{sec:construction}.
Let $\theta_0,\ldots,\theta_N$ be the cuts of this allocation, and output the vector $z$ given by~\eqref{eq:decoder}.
We show that $z$ is a $1/12$-approximate equilibrium, whatever the assignment.
Fix $v\in[s]$.
By Lemma~\ref{lem:decoding}, we have $|\theta_{a_v}-a_vB+12(Az)_v|\le\frac1{100}$.
If $(Az)_v<5/12$, then $\theta_{a_v}-a_vB+\frac{13}{2}>\frac{13}{2}-5-\frac1{100}>1$, and hence $z_v=1$.
If $(Az)_v>7/12$, then $\theta_{a_v}-a_vB+\frac{13}{2}<\frac{13}{2}-7+\frac1{100}<0$, and hence $z_v=0$.
Thus, $z$ is a $1/12$-approximate equilibrium.
By Theorem~\ref{thm:source}, this proves PPAD-hardness.
\end{proof}

Theorem~\ref{thm:main} also determines the complexity of exact envy-freeness.

\begin{corollary}\label{cor:exact}
Computing a contiguous envy-free allocation for normalized rational piecewise-constant additive valuations is PPAD-complete.
The hardness holds even when every density takes only the values $\lambda$ and $2\lambda$ for some $\lambda>0$ common to all agents and has at most six breakpoints.
\end{corollary}
\begin{proof}
Since an envy-free allocation is $\varepsilon$-envy-free for every $\varepsilon>0$, the hardness follows from Theorem~\ref{thm:main}, and membership is Lemma~\ref{lem:membership}.
\end{proof}

\section{Indivisible goods}\label{sec:discrete}
In this section, we prove Theorem~\ref{thm:discrete}, using the definitions of Section~\ref{sec:items} for goods.
For binary valuations, since each removed good lowers $\nu_i(X_j)$ by at most $1$, every EF$c$ allocation, and in particular every EF$c_{\mathrm{outer}}$ allocation, satisfies
\begin{equation}\label{eq:binaryEF}
 \nu_i(X_j)\le\nu_i(X_i)+c\qquad(i,j\in[n]).
\end{equation}

We prove hardness by discretizing the hard instances of Theorem~\ref{thm:main}, and membership by rounding exact envy-free divisions of a cake.

\subsection{Hardness}\label{sec:discrete-hardness}
The reduction places goods at the quantiles of the valuations.
Goldberg, Hollender, and Suksompong~\cite[Proposition~6.7]{GoldbergHS2020b} used the same construction for $c=0$, that is, for exact envy-freeness of the goods.

\begin{lemma}\label{lem:quantile}
For every fixed integer $c\ge0$, computing a contiguous EF$c$ allocation is PPAD-hard for binary valuations in which every good is valued by exactly one agent and all agents value the same number of goods.
The same holds for EF$c_{\mathrm{outer}}$.
\end{lemma}
\begin{proof}
Such allocations always exist, for $c=0$ by Marenco and Tetzlaff~\cite{MarencoT2014} since every good is valued by exactly one agent, and for $c\ge1$ by Igarashi~\cite{Igarashi2023}.
By Theorem~\ref{thm:main}, the \Cake{} problem is PPAD-hard even when $\varepsilon$ is encoded in unary and every density is positive.
We give a polynomial-time search reduction from such instances.
Consider an instance with valuations $V_1,\ldots,V_n$ and precision $\varepsilon$, and let $k=\lceil(c+2)/\varepsilon\rceil$.
Let $\Phi_i(t)=V_i([0,t])$ be the cumulative valuation of agent $i$.
For each $i\in[n]$ and $r\in[k]$, we place a good at the unique point $t_{ir}$ satisfying
\[
 \Phi_i(t_{ir})=\frac{r-1/2}{k},
\]
and let agent $i$ value this good at $1$ and all other agents at $0$.
Since the densities are positive and piecewise constant, all positions are rational and can be computed in polynomial time.
The instance consists of the $nk$ goods sorted by position, with ties broken arbitrarily.
Since $\varepsilon$ is given in unary and $c$ is fixed, $nk$ is polynomial in the input length.

Since every EF$c_{\mathrm{outer}}$ allocation is EF$c$, it suffices to map EF$c$ allocations.
For $0<\beta<nk$, let $\tau(\beta)$ be the midpoint of the positions of the $\beta$th and the $(\beta+1)$st good, and let $\tau(0)=0$ and $\tau(nk)=1$.
Write $D_i(\beta)$ for the number of goods valued by agent $i$ among the first $\beta$ goods.
The key estimate is
\[
 \left|\Phi_i(\tau(\beta))-\frac{D_i(\beta)}{k}\right|\le\frac1{2k}
 \qquad(0\le\beta\le nk).
\]
Indeed, let $d=D_i(\beta)$, $t_{i0}=0$, and $t_{i,k+1}=1$.
Since the goods are sorted by position and the goods of agent $i$ lie at distinct positions, the first $\beta$ goods contain exactly the first $d$ goods of agent $i$, and hence $t_{id}\le\tau(\beta)\le t_{i,d+1}$.
Since $\Phi_i(t_{id})\ge(d-1/2)/k$ and $\Phi_i(t_{i,d+1})\le(d+1/2)/k$, the estimate follows.

Given a contiguous EF$c$ allocation with indices $\beta_0,\ldots,\beta_n$ and a bijection $\sigma$, use the cake cuts $\tau(\beta_0),\ldots,\tau(\beta_n)$ and the same $\sigma$, and let $Y_j$ be the piece of agent $j$.
Subtracting the cumulative estimates at the two ends gives
\[
 \left|V_i(Y_j)-\frac{\nu_i(X_j)}{k}\right|\le\frac1k.
\]
Therefore, by~\eqref{eq:binaryEF},
\[
 V_i(Y_j)-V_i(Y_i)\le\frac{\nu_i(X_j)-\nu_i(X_i)+2}{k}\le\frac{c+2}{k}\le\varepsilon.
 \qedhere
\]
\end{proof}

\subsection{Membership}\label{sec:discrete-membership}
We prove membership by rounding an exact envy-free division of a cake with piecewise-linear virtual valuations.
We also use this argument for chores in Section~\ref{sec:chores}.

We identify the $\ell$th item with the interval $[\ell-1,\ell]$, so that cuts $0=y_0\le y_1\le\cdots\le y_n=m$ divide $[0,m]$ into the intervals $F_j=[y_{j-1},y_j]$.
For a set function $w$ on the items with $w(\emptyset)=0$ and a set $S$ of consecutive items, let $w^+(S)$ and $w^-(S)$ be the maximum and the minimum of $w(S\setminus G)$ over all $G\subseteq S$ with $|G|\le1$ such that $S\setminus G$ is connected or empty.
Dupr\'e la Tour and Igarashi~\cite[Theorem~1]{DupreLaTourI2026} constructed, for every such $w$, a function $\tilde w$ on the intervals of $[0,m]$, together with a map $\mathcal R$ from cut vectors to partitions $(A_1,\ldots,A_n)$ of the items into intervals, where $A_j$ is the $j$th interval from the left, such that
\begin{equation}\label{eq:rounding}
 w^-(A_j)\le\tilde w(F_j)\le w^+(A_j)
\end{equation}
for every cut vector and every $j\in[n]$.
The function $\tilde w$ is nonnegative if $w$ is nonnegative, nonpositive if $w$ is nonpositive, and $0$ on intervals of length zero.
It is the affine interpolation, on the Kuhn triangulation of $\{(a,b):0\le a\le b\le m\}$ with step $\frac13$, of its values at the grid points, each of which is the median of the values of $w$ on three sets of consecutive items~\cite{BiloLV2026}.
The map $\mathcal R$ does not depend on $w$: it locates a simplex of the Kuhn triangulation with step $\frac13$ of the set of cut vectors that contains the given cut vector and applies a fixed rule to its cuts.
For additive $w$, the value of $w$ on a set of consecutive items is a sum of item values, and hence both maps can be evaluated in polynomial time.

Our membership proofs use the following result of Filos-Ratsikas, Hansen, H{\o}gh, and Hollender~\cite{FilosRatsikasHHH2026}.
Consider a cake $[0,\Lambda]$ and $n$ agents, where agent $i$ values each interval $[a,b]\subseteq[0,\Lambda]$ at $U_i([a,b])$, given by a polynomial-size piecewise-linear arithmetic circuit with inputs $a$ and $b$.
Suppose also that $U_i([a,a])=0$, and that either $U_i([a,b])>0$ for all $i$ and $a<b$, or $U_i([a,b])<0$ for all $i$ and $a<b$.
Then, a contiguous allocation with $U_i(X_i)\ge U_i(X_j)$ for all $i,j$ can be computed in PPAD.
Filos-Ratsikas et al.\ assume nonnegative values, which we obtain in the second case by adding to all values a constant of polynomial bit length that bounds $|U_i|$, and this does not affect envy-freeness.

\begin{lemma}\label{lem:discrete-membership}
For nonnegative rational additive valuations of goods, and also for nonnegative rational additive costs of chores, computing a contiguous EF1$_{\mathrm{outer}}$ allocation lies in PPAD.
\end{lemma}
\begin{proof}
Scaling the values of each agent by a positive integer, we may assume that they are integers.
Let $w_i=\nu_i$ for goods and $w_i=-\kappa_i$ for chores.
We apply the result above to the cake $[0,m]$, in which agent $i$ values each interval $F$ at
\[
 U_i(F)=\tilde w_i(F)+\frac{|F|}{2m}\quad\text{for goods},\qquad
 U_i(F)=\tilde w_i(F)-\frac{|F|}{2m}\quad\text{for chores}.
\]
Since $\tilde w_i([a,a])=0$ and $\tilde w_i$ is nonnegative for goods and nonpositive for chores, these values are $0$ for intervals of length zero, and positive for goods and negative for chores otherwise.
Moreover, since $\tilde w_i$ is an affine interpolation of $O(m^2)$ values, each $U_i$ is given by a piecewise-linear arithmetic circuit of polynomial size.
Thus, we can compute in PPAD cuts $y$ and a bijection $\sigma$ with $U_i(F_{\sigma(i)})\ge U_i(F_j)$ for all $i,j$, and we output $\mathcal R(y)$ with the same $\sigma$.
Since the perturbation is at most $\frac12$, these inequalities and~\eqref{eq:rounding} give
\[
 w_i^+(A_{\sigma(i)})\ge\tilde w_i(F_{\sigma(i)})\ge\tilde w_i(F_j)-\tfrac12\ge w_i^-(A_j)-\tfrac12,
\]
and hence $w_i^+(A_{\sigma(i)})\ge w_i^-(A_j)$, since both sides are integers.
This is EF1$_{\mathrm{outer}}$: for goods, $\nu_i(A_{\sigma(i)})$ is at least the value of $A_j$ after removing at most one good at an end, and for chores, $\kappa_i(A_j)$ is at least the cost of $A_{\sigma(i)}$ after removing at most one chore at an end.
\end{proof}

\begin{proof}[Proof of Theorem~\ref{thm:discrete}]
Let $c\ge1$ be a fixed integer, or let $c=\lfloor m^\alpha\rfloor$ for a fixed rational $0<\alpha<1$.
Every EF1$_{\mathrm{outer}}$ allocation is both EF$c$ and EF$c_{\mathrm{outer}}$, and hence Lemma~\ref{lem:discrete-membership} gives totality and a search reduction to a problem in PPAD.
Both solution predicates are polynomial-time verifiable: for each pair of agents, remove up to $c$ goods of largest value from the envied bundle for EF$c$, or enumerate the pairs of a prefix and a suffix of total size at most $c$ for EF$c_{\mathrm{outer}}$.
Thus, both problems belong to PPAD.
For a fixed $c$, Lemma~\ref{lem:quantile} gives hardness under the stated restrictions on binary valuations.
For $c=\lfloor m^\alpha\rfloor$, take in the proof of Lemma~\ref{lem:quantile} the smallest integer $k\ge4/\varepsilon$ with $k^{1-\alpha}\ge2n^\alpha/\varepsilon$, which is polynomial in the input length.
Since $m=nk$, we have $(c+2)/k\le n^\alpha k^{\alpha-1}+2/k\le\varepsilon$, and hence the proof of Lemma~\ref{lem:quantile} applies.
\end{proof}

In contrast, for a fixed $\delta>0$ and $c=\lceil\delta m\rceil$, computing a contiguous EF$c$ or EF$c_{\mathrm{outer}}$ allocation is solvable in polynomial time.
If $n\ge1/\delta$, dividing the path into $n$ intervals of at most $\lceil m/n\rceil\le c$ goods and assigning them arbitrarily gives an EF$c_{\mathrm{outer}}$ allocation, since every bundle can be removed entirely.
Otherwise, $n$ is a constant, and we can enumerate all cuts and assignments.

\section{Extension to chores}\label{sec:chores}
Our results extend to chores, where each agent prefers pieces of small cost.
We use the definitions for chores in Sections~\ref{sec:cake} and~\ref{sec:items}.
Contiguous EF1$_{\mathrm{outer}}$ allocations of chores always exist~\cite{SunL2026,BiloLV2026}.
Bil\`o, Loebl, and Vinci~\cite{BiloLV2026} conjectured that computing contiguous EF1 allocations is PPAD-complete for goods and for chores, and similarly for a stronger notion that coincides with EF1$_{\mathrm{outer}}$ for additive valuations.
Theorem~\ref{thm:discrete} confirms these conjectures for additive valuations of goods, and the following proposition confirms them for additive costs.

\begin{proposition}\label{prop:chores}
Computing a contiguous $\varepsilon$-envy-free division of a chore for normalized rational piecewise-constant additive costs is PPAD-complete, and so is computing a contiguous envy-free one.
The hardness holds even when $\varepsilon$ is given in unary, and even when every cost density takes only the values $\lambda$ and $2\lambda$ for some $\lambda>0$ common to all agents and has at most six breakpoints.
For every fixed integer $c\ge1$, and also for $c=\lfloor m^\alpha\rfloor$ with every fixed rational $0<\alpha<1$, computing a contiguous EF$c$ or EF$c_{\mathrm{outer}}$ allocation of chores on a path is PPAD-complete for nonnegative rational additive costs, even for binary costs in which every chore has positive cost for exactly one agent and all agents have positive cost for the same number of chores.
\end{proposition}
\begin{proof}
\emph{Membership.} For chore division, let $W\ge3$ be an integer such that every numerator and every denominator of the rational numbers describing the costs is at most $W$ in absolute value, let $K\ge1$ be an integer such that every cost density is constant on at most $K$ intervals, and let $\delta=W^{-20Kn}$.
Applying the result stated before Lemma~\ref{lem:discrete-membership} to the values $-C_i(X)-\delta|X|$ of the intervals $X$, which are $0$ for intervals of length zero and negative otherwise, gives in PPAD a $\delta$-envy-free allocation for the costs.
Goldberg, Hollender, and Suksompong~\cite[Lemma~6.2]{GoldbergHS2020b} showed that, for piecewise-constant valuations, an envy-free allocation can be computed in polynomial time from any $W^{-20Kn}$-envy-free allocation by linear programming, and their proof applies to costs without change.
Hence, an envy-free allocation for the costs can be computed in PPAD, and it is also $\varepsilon$-envy-free.
For chores on a path, Lemma~\ref{lem:discrete-membership} gives an EF1$_{\mathrm{outer}}$ allocation, which is both EF$c$ and EF$c_{\mathrm{outer}}$, and both predicates can be verified in polynomial time as in the proof of Theorem~\ref{thm:discrete}.

\emph{Hardness.} Appendix~\ref{app:chores} adapts the reduction of Section~\ref{sec:hardness} to costs and proves hardness for chore division under the stated restrictions, which also gives hardness for exact envy-freeness.
Finally, the quantile transfer of Lemma~\ref{lem:quantile} applies to these instances without change, with costs in place of values and with $k$ chosen as in the proof of Theorem~\ref{thm:discrete}: for binary costs, every EF$c$ allocation satisfies $\kappa_i(X_i)\le\kappa_i(X_j)+c$, and hence $C_i(Y_i)-C_i(Y_j)\le(c+2)/k\le\varepsilon$.
\end{proof}

\section{Concluding remarks}\label{sec:conclusion}
Our hardness extends to \emph{piecewise-uniform} valuations, whose densities are a positive constant on a finite union of intervals and zero elsewhere.
Take a hard instance of Lemma~\ref{lem:quantile} for $c=0$, replace the $j$th of its $m$ goods by the interval $[(j-1)/m,j/m]$, and let each agent have a uniform density on the intervals of the goods it values.
By Goldberg, Hollender, and Suksompong~\cite[Proposition~6.8]{GoldbergHS2020b}, every contiguous $1/m$-envy-free allocation of this cake can be rounded in polynomial time to an envy-free allocation of the goods.
Hence, the problem remains PPAD-complete for piecewise-uniform valuations with pairwise disjoint supports, even when $\varepsilon$ is given in unary.
In contrast, the problem is easy when the supports are pairwise disjoint and each of them consists of at most two intervals.
Then, each agent $i$ values one of its support intervals, say $J_i$, at least $1/2$, and the intervals $J_1,\ldots,J_n$ are pairwise disjoint.
Placing one cut between each pair of consecutive intervals $J_i$ divides the whole cake into $n$ pieces, each containing exactly one of these intervals.
Giving each agent $i$ the piece containing $J_i$ yields an envy-free allocation, since every agent values its piece at least $1/2$, and hence at least as much as the rest of the cake.
It would be interesting to determine the complexity for pairwise disjoint supports of at most three intervals each, and for supports of at most two intervals that may overlap, the setting in which consensus halving is PPA-complete~\cite{FilosRatsikasHSZ2023}.

Several further questions remain open.
In the hard instances of Theorem~\ref{thm:main}, $\varepsilon=\Theta(n^{-4})$, whereas Barman and Kulkarni~\cite{BarmanK2023} computed a contiguous $(1/4+o(1))$-envy-free allocation in polynomial time.
For general normalized rational piecewise-constant additive valuations, the complexity for fixed $0<\varepsilon\le1/4$ remains open.

Under a density bound $\rho_i\le M$, a contiguous $\varepsilon$-envy-free allocation can be computed in time $2^{O(M/\varepsilon)}$ times a polynomial in the input length, and hence there is an efficient polynomial-time approximation scheme for additive envy when $M$ is a constant.
If $n\ge M/\varepsilon$, dividing the cake into $n$ equal-length pieces gives every piece value at most $\varepsilon$ to every agent, and hence any assignment is $\varepsilon$-envy-free.
Otherwise, let $Q=\lceil4M/\varepsilon\rceil$, and enumerate the at most $\binom{Q+n-1}{n-1}\le2^{Q+n}$ partitions into $n$ pieces whose cuts are multiples of $1/Q$.
For each partition, search for a perfect matching in the bipartite graph that joins agent $i$ to piece $j$ if agent $i$ values piece $j$ within $\varepsilon$ of its favorite piece.
Every such matching gives an $\varepsilon$-envy-free allocation, and some partition has one, since rounding down the cuts of a contiguous envy-free allocation to multiples of $1/Q$ changes every envy difference by at most $4M/Q\le\varepsilon$.
Theorem~\ref{thm:main} rules out a fully polynomial-time approximation scheme unless $\mathrm{PPAD}\subseteq\mathrm{FP}$, even with two density values and at most six breakpoints.
Closing the gap between hardness at $\varepsilon=\Theta(n^{-4})$ and the guarantee $\varepsilon=\Theta(1/n)$ of equal division remains open.

\appendix
\section{Hardness for chores}\label{app:chores}
We prove the hardness part of Proposition~\ref{prop:chores} for chore division by reusing the reduction of Section~\ref{sec:hardness} with one change in the sets $S_v$.

\paragraph{Construction.}
Under the costs defined below, every piece indexed by $S_v$ becomes longer, rather than shorter, by about $z_v/2$, and hence we reverse the transfers between the bins and double their size.
Let $R'=32(s+2)$ and $N'=(s+1)R'$, and construct the sets $S_v$ and the cyclic order $\pi$ as in Section~\ref{sec:construction}, with $R'$ and $N'$ in place of $R$ and $N$, except that, for every arc $(v,u)$, we transfer $24$ indices of $S_v$ from the $u$th bin to the $(u+1)$st bin.
Thus, the block for $S_v$ in the $u$th bin has $32-24A_{uv}+24A_{u-1,v}\ge8$ indices, the blocks in each bin have at most $32s+48=R'-16$ indices in total, \ref{S:arc} holds as before, and \ref{S:counts} becomes $|S_v|=32(s+1)$ and $|S_v\cap[a_u]|=32u-24A_{uv}$, where $a_u=uR'$.
Let $B=40N'$ and $L=N'B$, let $U_i$ be the valuations of Section~\ref{sec:construction} for these sets, and let agent $i$ have the cost $C_i(F)=3|F|-U_i(F)$.
Its density is $3$ minus that of $U_i$, and hence it takes only the values $1$ and $2$ and has at most six breakpoints, and $C_i([0,L])=2L-6$.
The target instance consists of the normalized costs and $\varepsilon=1/(8000N'^4)$, and its cost densities take only the values $\lambda$ and $2\lambda$ for $\lambda=L/(2L-6)$.
As in Section~\ref{sec:construction}, every $\varepsilon$-envy-free allocation of it gives an allocation of $[0,L]$ with $C_i(P_{\sigma(i)})\le C_i(P_j)+\eta$ for all $i,j$, where $\eta=1/(100N'^2)$.

\paragraph{Analysis.}
Let $W_i=-C_i$, so that $W_i(F)=U_i(F)-3|F|$ lies between $-2|F|$ and $-2|F|+6$, and the allocation satisfies $W_i(P_{\sigma(i)})\ge W_i(P_j)-\eta$.
If agent $i$ receives a longest piece and $P_j$ is a shortest piece, then $-2\ell_{\sigma(i)}+6\ge W_i(P_{\sigma(i)})\ge W_i(P_j)-\eta\ge-2\ell_j-\eta$.
Hence, the lengths differ by less than $4$, every cut satisfies $|\theta_j-jB|<4N'$, and Lemma~\ref{lem:localization} holds with the same proof.
Lemma~\ref{lem:envelope} also holds, since its proof uses only Lemma~\ref{lem:localization}, \ref{S:arc}, and the bonuses.
Subtracting $3\ell_j$ from both $U_i(P_j)$ and $q_j$ shows that $W_i$ satisfies~\eqref{eq:envelope} with the prices $q'_j=q_j-3\ell_j=h_j-2\ell_j$, where $h_j=z_v$ for $j\in S_v$ and $h_j=0$ for $j$ in no set.
Thus, the argument after~\eqref{eq:envelope} and Lemma~\ref{lem:cycle} give $\max_jq'_j-\min_jq'_j\le N'\eta$.
Finally, let $\bar q'$ be the average of the $q'_j$.
Since $\ell_j=(h_j-q'_j)/2$, $\sum_j\ell_j=N'B$, and $\sum_jh_j=32(s+1)\sum_vz_v$, we have $\bar q'+2B=\frac{32(s+1)}{N'}\sum_vz_v$.
Since $\sum_{j\in[a_u]}h_j=\sum_v|S_v\cap[a_u]|\,z_v$ and $a_u\cdot32(s+1)/N'=32u$, we obtain
\begin{align*}
 \theta_{a_u}-a_uB
 &=\sum_{j\in[a_u]}(\ell_j-B)
 =-\frac12\sum_{j\in[a_u]}(q'_j-\bar q')-\frac{a_u}2(\bar q'+2B)+\frac12\sum_{j\in[a_u]}h_j\\
 &=-\frac12\sum_{j\in[a_u]}(q'_j-\bar q')+\frac12\sum_{v\in[s]}\bigl(|S_v\cap[a_u]|-32u\bigr)z_v
 =-\frac12\sum_{j\in[a_u]}(q'_j-\bar q')-12(Az)_u.
\end{align*}
Hence, $|\theta_{a_u}-a_uB+12(Az)_u|\le N'^2\eta/2=\frac1{200}$, and the proof of Theorem~\ref{thm:main} then applies without change.

\section*{AI Methodology}
An initial version of the reduction in Section~\ref{sec:hardness} was obtained with OpenAI's ChatGPT.
The authors then verified, corrected, and simplified the construction and the proofs, using OpenAI's Codex and Anthropic's Claude to assist with developing research ideas, constructing and checking proofs, and drafting and revising the manuscript.
The authors take full responsibility for the content of this paper, including all mathematical claims and proofs.

{\small
\bibliographystyle{alphaurl}
\bibliography{references}

@article{Su1999,
  author = {Francis Edward Su},
  title = {Rental Harmony: {Sperner}'s Lemma in Fair Division},
  journal = {The American Mathematical Monthly},
  volume = {106},
  number = {10},
  pages = {930--942},
  year = {1999},
  doi = {10.2307/2589747},
}

@article{DengQS2012,
  author = {Xiaotie Deng and Qi Qi and Amin Saberi},
  title = {Algorithmic Solutions for Envy-Free Cake Cutting},
  journal = {Operations Research},
  volume = {60},
  number = {6},
  pages = {1461--1476},
  year = {2012},
  doi = {10.1287/opre.1120.1116},
}

@article{GoldbergHS2020b,
  author = {Paul W. Goldberg and Alexandros Hollender and Warut Suksompong},
  title = {Contiguous Cake Cutting: Hardness Results and Approximation Algorithms},
  journal = {Journal of Artificial Intelligence Research},
  volume = {69},
  pages = {109--141},
  year = {2020},
  doi = {10.1613/jair.1.12222},
}

@inproceedings{Alijani+2017,
  author = {Reza Alijani and Majid Farhadi and Mohammad Ghodsi and Masoud Seddighin and Ahmad S. Tajik},
  title = {Envy-Free Mechanisms with Minimum Number of Cuts},
  booktitle = {Proceedings of the 31st AAAI Conference on Artificial Intelligence (AAAI 2017)},
  pages = {312--318},
  year = {2017},
  doi = {10.1609/aaai.v31i1.10584},
}

@inproceedings{BarmanK2023,
  author = {Siddharth Barman and Pooja Kulkarni},
  title = {Approximation Algorithms for Envy-Free Cake Division with Connected Pieces},
  booktitle = {Proceedings of the 50th International Colloquium on Automata, Languages, and Programming (ICALP 2023)},
  pages = {16:1--16:19},
  year = {2023},
  doi = {10.4230/LIPIcs.ICALP.2023.16},
}

@article{HollenderR2025,
  author = {Alexandros Hollender and Aviad Rubinstein},
  title = {Envy-Free Cake-Cutting for Four Agents},
  journal = {Journal of the ACM},
  volume = {72},
  number = {5},
  pages = {34:1--34:54},
  year = {2025},
  doi = {10.1145/3765615},
}

@misc{HollenderMR2026,
  author = {Alexandros Hollender and Gilbert Maystre and Kilian Risse},
  title = {On Cutting Cakes and Crossing Curves},
  year = {2026},
  note = {arXiv:2606.13980. Accepted at EC 2026},
}

@article{DeligkasFHM2024,
  author = {Argyrios Deligkas and John Fearnley and Alexandros Hollender and Themistoklis Melissourgos},
  title = {{Pure-Circuit}: Tight Inapproximability for {PPAD}},
  journal = {Journal of the ACM},
  volume = {71},
  number = {5},
  pages = {31:1--31:48},
  year = {2024},
  doi = {10.1145/3678166},
}

@article{Papadimitriou1994,
  author = {Christos H. Papadimitriou},
  title = {On the Complexity of the Parity Argument and Other Inefficient Proofs of Existence},
  journal = {Journal of Computer and System Sciences},
  volume = {48},
  number = {3},
  pages = {498--532},
  year = {1994},
  doi = {10.1016/S0022-0000(05)80063-7},
}

@inproceedings{Bilo+2019,
  author = {Vittorio Bil\`{o} and Ioannis Caragiannis and Michele Flammini and Ayumi Igarashi and Gianpiero Monaco and Dominik Peters and Cosimo Vinci and William S. Zwicker},
  title = {Almost Envy-Free Allocations with Connected Bundles},
  booktitle = {Proceedings of the 10th Innovations in Theoretical Computer Science Conference (ITCS 2019)},
  pages = {14:1--14:21},
  year = {2019},
  doi = {10.4230/LIPIcs.ITCS.2019.14},
}

@inproceedings{Igarashi2023,
  author = {Ayumi Igarashi},
  title = {How to Cut a Discrete Cake Fairly},
  booktitle = {Proceedings of the 37th AAAI Conference on Artificial Intelligence (AAAI 2023)},
  pages = {5681--5688},
  year = {2023},
  doi = {10.1609/aaai.v37i5.25705},
}

@article{KawaseRS2026,
  author = {Yasushi Kawase and Bodhayan Roy and Mohammad Azharuddin Sanpui},
  title = {Contiguous Allocation of Indivisible Items on a Path},
  journal = {Autonomous Agents and Multi-Agent Systems},
  volume = {40},
  pages = {44},
  year = {2026},
  doi = {10.1007/s10458-026-09773-5},
}

@article{ChenDT2009,
  author = {Xi Chen and Xiaotie Deng and Shang-Hua Teng},
  title = {Settling the Complexity of Computing Two-Player {Nash} Equilibria},
  journal = {Journal of the ACM},
  volume = {56},
  number = {3},
  pages = {14:1--14:57},
  year = {2009},
  doi = {10.1145/1516512.1516516},
}

@article{Stromquist1980,
  author = {Walter Stromquist},
  title = {How to Cut a Cake Fairly},
  journal = {The American Mathematical Monthly},
  volume = {87},
  number = {8},
  pages = {640--644},
  year = {1980}
}

@article{Stromquist2008,
  author = {Walter Stromquist},
  title = {Envy-Free Cake Divisions Cannot Be Found by Finite Protocols},
  journal = {The Electronic Journal of Combinatorics},
  volume = {15},
  pages = {\#R11},
  year = {2008}
}

@article{MarencoT2014,
  author = {Javier Marenco and Tom\'{a}s Tetzlaff},
  title = {Envy-Free Division of Discrete Cakes},
  journal = {Discrete Applied Mathematics},
  volume = {164},
  pages = {527--531},
  year = {2014}
}

@inproceedings{GaoRRS2024,
  author = {Ruiquan Gao and Mohammad Roghani and Aviad Rubinstein and Amin Saberi},
  title = {Hardness of Approximate {Sperner} and Applications to Envy-Free Cake Cutting},
  booktitle = {Proceedings of the 65th IEEE Annual Symposium on Foundations of Computer Science (FOCS 2024)},
  pages = {1294--1331},
  year = {2024},
  doi = {10.1109/FOCS61266.2024.00085},
}

@inproceedings{PapadimitriouP2021,
  author = {Christos H. Papadimitriou and Binghui Peng},
  title = {Public Goods Games in Directed Networks},
  booktitle = {Proceedings of the 22nd ACM Conference on Economics and Computation (EC 2021)},
  pages = {745--762},
  year = {2021},
  doi = {10.1145/3465456.3467616},
}

@inproceedings{ChenCPY2022,
  author = {Thomas Chen and Xi Chen and Binghui Peng and Mihalis Yannakakis},
  title = {Computational Hardness of the {Hylland}--{Zeckhauser} Scheme},
  booktitle = {Proceedings of the 33rd Annual ACM-SIAM Symposium on Discrete Algorithms (SODA 2022)},
  pages = {2253--2268},
  year = {2022},
  doi = {10.1137/1.9781611977073.90}
}

@inproceedings{FilosRatsikasG2018,
  author = {Aris Filos-Ratsikas and Paul W. Goldberg},
  title = {Consensus Halving is {PPA}-Complete},
  booktitle = {Proceedings of the 50th Annual ACM Symposium on Theory of Computing (STOC 2018)},
  pages = {51--64},
  year = {2018},
  doi = {10.1145/3188745.3188880}
}

@article{FilosRatsikasHSZ2023,
  author = {Aris Filos-Ratsikas and Alexandros Hollender and Katerina Sotiraki and Manolis Zampetakis},
  title = {Consensus-Halving: Does It Ever Get Easier?},
  journal = {SIAM Journal on Computing},
  volume = {52},
  number = {2},
  pages = {563--602},
  year = {2023},
  doi = {10.1137/20M1387493},
}

@article{Suksompong2019,
  author = {Warut Suksompong},
  title = {Fairly Allocating Contiguous Blocks of Indivisible Items},
  journal = {Discrete Applied Mathematics},
  volume = {260},
  pages = {227--236},
  year = {2019},
  doi = {10.1016/j.dam.2019.01.036}
}

@article{FilosRatsikasHHH2026,
  author = {Aris Filos-Ratsikas and Kristoffer Arnsfelt Hansen and Kasper H{\o}gh and Alexandros Hollender},
  title = {{PPAD}-Membership for Problems with Exact Rational Solutions: A General Approach via Convex Optimization},
  journal = {SIAM Journal on Computing},
  pages = {STOC24-297--STOC24-421},
  year = {2026},
  note = {Preliminary version in STOC 2024. Full version: arXiv:2312.01237},
  doi = {10.1137/24M1689089},
}

@inproceedings{SunL2026,
  author = {Ankang Sun and Bo Li},
  title = {Connected {EF1} Allocations Exist in Discrete Chore Cutting},
  booktitle = {Proceedings of the 35th International Joint Conference on Artificial Intelligence (IJCAI-ECAI 2026)},
  pages = {3532--3540},
  year = {2026},
  doi = {10.24963/ijcai.2026/393},
}

@misc{DupreLaTourI2026,
  author = {Max {Dupr{\'e} la Tour} and Ayumi Igarashi},
  title = {From Cake-Cutting and Necklace-Splitting to Fair Division of Indivisible Items},
  year = {2026},
  note = {arXiv:2608.04340},
}

@inproceedings{BiloLV2026,
  author = {Vittorio Bil{\`o} and Martin Loebl and Cosimo Vinci},
  title = {Approximately Envy-Free and Equitable Allocations of Indivisible Items for Non-Monotone Valuations},
  booktitle = {Proceedings of the 40th AAAI Conference on Artificial Intelligence (AAAI 2026)},
  pages = {16700--16708},
  year = {2026},
  doi = {10.1609/aaai.v40i20.38712},
}

@article{DaskalakisGP2009,
  author = {Constantinos Daskalakis and Paul W. Goldberg and Christos H. Papadimitriou},
  title = {The Complexity of Computing a {Nash} Equilibrium},
  journal = {SIAM Journal on Computing},
  volume = {39},
  number = {1},
  pages = {195--259},
  year = {2009},
  doi = {10.1137/070699652},
}

@inproceedings{FilosRatsikasG2019,
  author = {Aris Filos-Ratsikas and Paul W. Goldberg},
  title = {The Complexity of Splitting Necklaces and Bisecting Ham Sandwiches},
  booktitle = {Proceedings of the 51st Annual ACM Symposium on Theory of Computing (STOC 2019)},
  pages = {638--649},
  year = {2019},
  doi = {10.1145/3313276.3316334},
}

@article{DeligkasFHM2025,
  author = {Argyrios Deligkas and John Fearnley and Alexandros Hollender and Themistoklis Melissourgos},
  title = {Constant Inapproximability for {PPA}},
  journal = {SIAM Journal on Computing},
  volume = {54},
  number = {1},
  pages = {163--192},
  year = {2025},
  doi = {10.1137/22M1536613},
}
}
\end{document}